\documentclass[a4paper,UKenglish]{lipics}
\usepackage{ifthen}
\newcommand{\techRep}{true} 
\newcommand{\iftechrep}{\ifthenelse{\equal{\techRep}{true}}}

\usepackage{amsfonts,amstext,amsmath,amssymb,stmaryrd,mathrsfs,calc,amsthm,xspace,mathtools}
\usepackage{bm}
\usepackage{comment}
\usepackage{framed}
\usepackage{graphics}
\usepackage{enumerate}
\usepackage{todonotes}

\usepackage{mathtools}

\usepackage{xspace,tabularx,multirow}
\usepackage{color}
\usepackage{colortbl}

\usepackage{macros}
\usepackage{hyperref}

\usepackage{tikz}
\usepackage[normalem]{ulem}
\usepackage[capitalise,noabbrev,nameinlink]{cleveref}
\usepackage{thm-restate}

\usepackage{algorithm}
\usepackage[noend]{algpseudocode}

\makeatletter
\let\c@theorem\@undefined
\let\theorem\@undefined
\let\endtheorem\@undefined
\let\lemma\@undefined
\let\endlemma\@undefined
\let\corollary\@undefined
\let\endcorollary\@undefined
\let\definition\@undefined
\let\enddefinition\@undefined
\let\example\@undefined
\let\endexample\@undefined
\let\remark\@undefined
\let\endremark\@undefined
\makeatother
\theoremstyle{plain}
\newtheorem{theorem}{Theorem}
\newtheorem{lemma}[theorem]{Lemma}
\newtheorem{corollary}[theorem]{Corollary}
\theoremstyle{definition}
\newtheorem{definition}[theorem]{Definition}
\newtheorem{example}[theorem]{Example}
\theoremstyle{remark}
\newtheorem{remark}[theorem]{Remark}
\crefname{lemma}{Lemma}{Lemmas}
\crefname{definition}{Definition}{Definitions}
\theoremstyle{plain}
\newtheorem{proposition}[theorem]{Proposition}
\crefname{proposition}{Proposition}{Propositions}
\theoremstyle{remark}
{\itshape}{\rmfamily}
\crefname{fact}{Fact}{Facts}

\crefname{claim}{Claim}{Claims}

\usetikzlibrary{arrows, automata,calc,shapes.arrows, backgrounds, fadings,intersections,  decorations.markings, positioning, snakes}
\usetikzlibrary{shadows}
\tikzfading[name=fade out, inner color=transparent!0,
  outer color=transparent!100]
\tikzset{
    every picture/.style={>=stealth,auto,node distance=2cm,}
}
\tikzstyle{every state}=[
    inner sep=0pt,
    minimum size=0.7cm,
    fill= white,
    fill opacity=0.2,
    draw= black,
    text= black,
    text opacity=1,
    scale=0.5,
    ]

\title{A Polynomial-Time Algorithm for Reachability in Branching VASS
  in Dimension One} 
\titlerunning{Reachability in Branching VASS in
  Dimension One}

\author[1]{Stefan G\"oller\thanks{Supported by Labex Digicosme,
    Univ.\ Paris-Saclay, project VERICONISS.}}
\author[1]{Christoph Haase\textsuperscript{*}} 
\affil[1]{LSV, CNRS \& ENS Cachan\\ Universit\'e Paris-Saclay,
  France\\ \texttt{\{goeller,haase\}@lsv.ens-cachan.fr}}

\author[2]{Ranko Lazi\'c\thanks{Supported by the EPSRC, grants EP/M011801/1 and EP/M027651/1.}}
\author[2]{Patrick Totzke\textsuperscript{$\dagger$}}
\affil[2]{DIMAP, Department of Computer Science\\ University of Warwick, United Kingdom\\\texttt{\{r.s.lazic,p.totzke\}@warwick.ac.uk}}

\authorrunning{S.\ G\"oller, C.\ Haase, R.\ Lazi\'c and P.\ Totzke}

\Copyright{Stefan G\"oller, Christoph Haase, Ranko Lazi\'c and Patrick Totzke}

\subjclass{F.1.1 Models of Computation}
\keywords{branching vector addition systems, reachability, coverability, boundedness}

\begin{document}

\maketitle

\begin{abstract}
  Branching VASS (\BVASS{}) generalise vector addition systems with
  states by allowing for special branching transitions that can
  non-deterministically distribute a counter value between two control
  states. A run of a BVASS consequently becomes a tree, and
  reachability is to decide whether a given configuration is the root
  of a reachability tree. This paper shows \P-completeness of
  reachability in \BVASS\ in dimension one, the first decidability
  result for reachability in a subclass of \BVASS\ known so
  far. Moreover, we show that coverability and boundedness in
  \BVASS\ in dimension one are \P-complete as well.
\end{abstract}

\section{Introduction}

Vector addition systems with states (VASS), equivalently known as
Petri nets, are a fundamental model of computation which comprise a
finite-state controller with a finite number of counters ranging over
the naturals. The number of counters is usually refereed to as the
dimension of the VASS. A configuration $q(\vec{n})$ of a VASS in
dimension $d$ consists of a control state $q$ and a valuation
$\vec{n}\in \N^d$ of the counters. A transition of a VASS can
increment and decrement counters and is enabled in a configuration
whenever the resulting counter values are all non-negative, otherwise
the transition is disabled. Consequently, VASS induce an infinite
transition system. Three of the most fundamental decision problems for
VASS are reachability, coverability and boundedness. Given a target
configuration $q(\vec{n})$ and some initial configuration,
reachability is to decide whether starting in the initial
configuration there exists a path ending in $q(\vec{n})$ in the
induced infinite transition system. Coverability asks whether some
configuration $q(\vec{n}')$ can be reached for some $\vec{n}'\ge
\vec{n}$, where ${\ge}$ is defined component-wise. Boundedness is the
problem to decide whether there are infinitely many different configurations
reachable from a given starting configuration.
Those decision problems find
a plethora of applications, for instance in the verification of
concurrent programs. Coverability can, for example, be used in order
to validate mutual exclusion properties of shared-memory concurrent
programs~\cite{GS92}; reachability is a key underlying decision
problem in the verification of liveness properties of finite-data
asynchronous programs~\cite{GM12}. Even though the complexity of
coverability and boundedness are well-understood and known to
be \EXPSPACE-complete~\cite{L76,Rack78}, the precise complexity of
reachability remains a major unsolved problem; a non-primitive recursive 
upper bound ($\mathbf{F}_{\omega^3}$) has only recently been
established~\cite{LS15} and the best known lower bound
is \EXPSPACE~\cite{L76}.

The situation is even more dissatisfying when considering
\emph{branching extensions} of VASS. Such \emph{branching VASS
  (BVASS)} are additionally equipped with special branching
transitions of the form $(q,p,p')$. When in a configuration
$q(\vec{n})$, a BVASS can simultaneously non-deterministically branch
into configurations $p(\vec{m})$ and $p'(\vec{m}')$ such that
$\vec{n}=\vec{m}+\vec{m}'$. Reachability of a configuration
$q(\vec{n})$ then is to decide whether there exists a proof tree whose
root is labelled with $q(\vec{n})$ and whose leaves are all labelled
with designated target control states in which all counters have value
zero; coverability and boundedness are defined analogously as
above. While coverability and boundedness are known to be
\textsf{2}-\EXPTIME-complete~\cite{DJLL13}, reachability in
BVASS is not known to be decidable, not even in any fixed dimension.
Recently, non-elementary lower bounds for reachability in BVASS have
been obtained~\cite{LaS15}. Reachability in BVASS is closely related
and in fact equivalent to decidability of the
multiplicative-exponential fragment of linear logic~\cite{GGS04}, and
also an underlying decision problem in various other applications for
instance in computational linguistics, cryptographic protocol
verification, data logics and concurrent program verification;
see~\cite{LaS15} for more details.

The primary contribution of this paper is to provide a polynomial-time
algorithm for reachability in \BVASS\ in dimension one (\BVASSone) and
to show that reachability is in fact \P-complete.  To the best of our
knowledge, we give the first decidability result for reachability in a
fragment of BVASS. Let us remark that a decidability result, in
particular with such low complexity is actually quite surprising. On
the one hand, due to the infinite state space of \BVASSone\ it is not
immediate that reachability is decidable. In particular, the emptiness
problem for conjunctive grammars over a unary alphabet, which can be
seen as a slight generalisation of \BVASSone\ with special alternating
transitions that can simultaneously branch into two control states
while retaining the same counter value (known as ABVASS$_1$), is
undecidable~\cite{JO10}. On the other hand, if we disallow branching
rules in ABVASS$_1$ and thus obtain AVASS$_1$ then reachability
is \PSPACE-complete~\cite{Ser06,JS07}.

Due to the presence of only one single counter, it is possible to
establish a small-model property and to show that if a configuration
is reachable in a \BVASSone\ then there exists a so-called
reachability tree of exponential size. What causes a main challenge
when establishing a polynomial-time algorithm is that this bound is
optimal in the sense that, as we show in \cref{sec:lower}, there exist
families of \BVASSone\ whose reachability trees are inherently of
exponential size, and which also contain an exponential number of
different counter values. Consequently, reachability cannot be
witnessed in polynomial time by explicitly constructing a witnessing
reachability tree. Instead, in \cref{sec:reachability} we show that
polynomial-time computable certificates for the reachability of a
configuration suffice. These certificates have two parts: the first is
a table that, for certain $d>0$ contains those pairs of control states
$q$ and residue classes $r$ modulo $d$ such that $q(n)$ is reachable
for some sufficiently large $n$ with $n\equiv r\bmod d$. This is
called residue reachability and described in \cref{sec:residue}. The
second part, described in \cref{ssec:expandable}, is a compressed
collection of incomplete small reachability trees, so-called
expandable partial reachability trees, whose leaves are either
accepting configurations or
have some ancestor node with the same control state
and a strictly smaller counter.
In the latter case, the corresponding subtree
can be repeated arbitrarily often, which leaves some 
configuration with an
arbitrarily large counter value in a certain residue class. This
eventually enables us to witness the existence of a reachability tree
via residue reachability.

In \cref{sec:boundedness}, we show that coverability and boundedness
are \P-complete for \BVASSone.  For coverability, the upper bound
follows easily via a reduction to reachability.  For boundedness, this
is not the case and we require a specifically tailored argument.

\iftechrep{Due to space constraints, the proofs of some statements can
  be found in an appendix.}{Due to space constraints, the proofs of
  some statements can be found in the technical report accompanying
  this article~\cite{GHLT16}.}

\section{Preliminaries}\label{sec:preliminaries}
We write $\Z$ and $\N$ for the sets of integers and non-negative integers,
respectively, and
define $[i,j]\defeq \{i,i+1,\ldots,j-1,j\}$, for given integers $i<j$.
For $d\geq 1$ we define $\Z_d\defeq[0,d-1]$.

The set of finite words over alphabet $A$ is denoted by $A^*$
and the length of a word $w\in A^*$ is written as $\len{w}$.
For two words $u,v\in A^*$, we say $u$ is a \emph{prefix} of $v$
(written as $u\prefixeq v$) if $v=uw$ for some $w\in A^*$.
It is a \emph{strict prefix} ($u\prefix v$) if $u\prefixeq v$ and $u\neq v$.
We say $u$ and $v$ are \emph{incomparable} if neither $u\preceq v$ nor $v\preceq u$.
A set $U\subseteq A^*$ is {\em prefix-closed} if for all $u\in U$
and all $v\in A^*$ we have that $v\prefixeq u$ implies $v\in U$.

Let $\Sigma$ be a set. A {\em $\Sigma$-labelled (finite) tree} is a
mapping $T \colon U\rightarrow \Sigma$ where $U\subseteq A^*$ is a
non-empty finite prefix-closed set of {\em nodes} for some finite set
$A$. For $V\subseteq U$, we define $T(V) \defeq \{ T(v) \mid v\in V
\}$. A {\em leaf} of $T$ is a node $u\in U$ such that there is no
$v\in U$ with $u\prec v$; every node of $T$ that is not a leaf is
called {\em inner node}. A node $u$ is an {\em ancestor}
(resp.\ \emph{descendant}) of a node $v$ if $u\preceq v$
(resp.\ $v\preceq u$) and a {\em strict ancestor} (resp.\ {\em strict
  descendant}) if $u\prec v$ (resp.\ $v\prec u$).
For any node $u$ we define the {\em subtree of $T$ rooted at $u$} as
$T^{\downarrow u}\colon u^{-1}U\rightarrow\Sigma$, where
$u^{-1}U\defeq\{x\in A^*\mid ux\in U\}$ and $T^{\downarrow u}(x)\defeq
T(ux)$.  Note that $u^{-1}U$ is a prefix-closed subset of $A$. We
define $h(u)\defeq\max\{|x| \mid x\in u^{-1}U\}$ to be the {\em
  height} of the subtree rooted at $u$ and and define $h(T)\defeq
h(\varepsilon)$.  Note that $h(u)=0$ if, and only if, $u$ is a
leaf. We say $T$ is {\em binary} if $U\subseteq\{0,1\}^*$; in this
case if for some node $u\in U$ we have that $u0\in U$, then $u0$ the
{\em left child} of $u$ and if $u1\in U$ we say that $u1$ is the {\em
  right child} of $u$.

\subsection{Branching Vector Addition Systems}
In the following, $\vec{n}$ and $\vec{z}$ will denote elements from
$\N^k$ and $\Z^k$, respectively; addition on $\Z^k$ is defined
component-wise.
\begin{definition}
  Let $k\geq 1$. A {\em $k$-dimensional branching vector addition
    system with states (BVASS$_k$)} is a tuple $\B=(Q,\Delta,F)$ where
  $Q$ is a finite set of {\em control states}, $\Delta\subseteq
  Q^3\cup (Q\times\{-1,0,1\}^k\times Q)$ is a finite set of
  \emph{transitions}, and $F\subseteq Q$ is a set of {\em final
    states}. The \emph{size} $\abs{\B}$ of a BVASS is defined as
  $\abs{\B}\defeq \abs{Q}+k \cdot \abs{\Delta}$.
\end{definition}

The semantics of BVASS is given in terms of reachability trees. A {\em
  partial reachability tree} of a \BVASS{k} $\B$ is a
$Q\times\N^k$-labelled binary tree $T \colon U\rightarrow
Q\times\N^k$, where each inner node $u\in U$ with $T(u)=(q,\vec{n})$
satisfies exactly one of the following conditions:
\begin{itemize}
\item $u0,u1\in U$, and if $T(u0)=(p,\vec{n}_0)$ and
  $T(u1)=(p',\vec{n}_1)$, then $\vec{n}=\vec{n}_0+\vec{n}_1$ and
  $(q,p,p')\in\Delta$; or
\item $u0\in U, u1\not\in U$, and if $T(u0)=(p,\vec{n}_0)$, then
  $\vec{n}_0=\vec{n}+\vec{z}$ and $(q,\vec{z},p)\in\Delta$.
\end{itemize}
Note that in the second condition, counter values can be seen as being
propagated top down. A {\em reachability tree} is a partial
reachability tree $T$ where $T(u)\in F\times\{0\}^k$ for all leaves
$u$ of $T$.  We call these nodes \emph{accepting} nodes.  For each
$j\in\N$ we say that a partial reachability tree $T$ is {\em
  $j$-bounded } if $T(u)\in Q\times[0,j]^k$ for all $u\in U$.  We call
$Q\times\N^k$ the set of {\em configurations} of $\B$ and for the sake
of readability often write its elements $(q,\vec{n})$ as
$q(\vec{n})$. We say that a configuration $q(\vec{n})$ is {\em
  reachable } if there exists a reachability tree $T$ with
$T(\varepsilon)=q(\vec{n})$. Note that in particular every
configuration in $F\times\{0\}^k$ is reachable. The reachability set
$\reach(q)$ of a control state $q$ is defined as $\reach(q) \defeq \{
\vec{n}\in \N \mid q(\vec{n}) \text{ is reachable} \}$.
The decision problem that we mainly focus on in this paper is
\emph{reachability}, defined as follows:
\medskip
\problemx{Reachability in BVASS$_k$} 
{A BVASS$_k$ $\B=(Q,\Delta,F)$, a control state $q$ and $\vec n\in \N^k$
  encoded in unary.}
{Is $q(\vec n)$ reachable?}
\medskip
Our main result is that reachability is \P-complete in dimension one.
\begin{theorem}\label{thm:main}
  Reachability in \BVASSone\ is \P-complete.
\end{theorem}

\section{Lower Bounds}\label{sec:lower}
As a warm-up exercise and in order to familiarise ourselves with
\BVASSone, we begin with proving a couple of lower bounds for the
reachability problem. First, it is not difficult to see that the
reachability problem is \P-hard via a reduction from the monotone
circuit value problem (\textsc{MCVP})~\cite{Pap94}. By simulating
$\vee$-gates of a Boolean by non-deterministic branching and
$\wedge$-gates by splitting transitions, the following statement can
easily be obtained.

\begin{restatable}{proposition}{propMCVPConstruction}\label{prop:mcvp-construction}
  Let $\C$ be a Boolean circuit. There exists a logspace computable
  \BVASSone\ $\B$ with a control state $q$ such that $q(0)$ is
  reachable if, and only if, $\C$ evaluates to true.
\end{restatable}

\begin{figure}
  \begin{center}
    \includegraphics[scale=1]{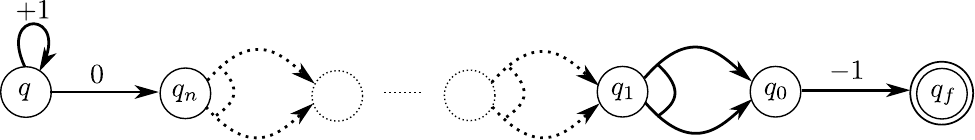}
  \end{center}
  \caption{Illustration of the \BVASSone\ $\B_n$. The reachability set
    of the control state $q_n$ is the singleton set $\{ 2^n \}$, and a
    reachability tree for $q(0)$ contains all counter values between
    $0$ and $2^n$.}
  \label{fig:gadget}
\end{figure}
A challenging aspect when providing a polynomial-time upper bound for
reachability in \BVASSone\ is that reachability trees may be of
exponential size and may contain an exponential number of nodes
labelled with distinct counter values. To see this, consider the
family $(\B_n)_{n\ge 0}$ of \BVASSone, where $\B_n\defeq
(Q_n,\Delta_n,F)$ and where $Q_n\defeq \{ q, q_f \} \cup \{
q_0,\ldots, q_n\}$, $\Delta_n \defeq \{ (q,+1,q), (q,0,q_n) \} \cup \{
(q_i,q_{i-1},q_{i-1}) \mid 0<i\le n \} \cup \{ (q_0,-1,q_f) \}$ and
$F\defeq \{ q_f \}$. The construction is illustrated in
\cref{fig:gadget}. It is easily seen that $q_i(N)$ is reachable if,
and only if, $N=2^i$. Observe that $\reach(q)=\{ 0,\ldots, 2^n \}$ is
finite and that the reachability tree of $q(0)$ contains all counter
values between $0$ and $2^n$. In particular, this allows us to obtain
the following hardness result in which the updates of the
\BVASSone\ are from $\{-1,0,+1 \}$ (i.e. encoded in unary), but the
initial configuration is given in binary, via a straight-forward
reduction from the \NP-complete \textsc{Subset Sum}
problem~\cite{Pap94}.
\begin{restatable}{proposition}{propNPHardness}\label{prop:np-hardness}
  Reachability in \BVASSone\ is \NP-hard if the initial configuration
  $q(n)$ is given in binary.
\end{restatable}

It is worth mentioning that the previous lemma enables us to derive as
a corollary an \NP-lower bound for reachability in \BVASS{2}. This is
in contrast to VASS where there is no difference between the
\NL-completeness of reachability in dimensions one and
two~\cite{VP75,ELT16}.
\begin{restatable}{corollary}{corNPBVASSTwo}
  Reachability in \BVASS{2} is \NP-hard.
\end{restatable}

\section{Reachability in \BVASSone}\label{sec:reachability}
Here, we show that reachability in \BVASSone\ is decidable in
polynomial time, thereby establishing the \P\ upper bound claimed in
\cref{thm:main}. In the first part, we consider a variation of the
reachability problem in which we are only interested in reaching
configurations that are sufficiently large and lie in a certain
residue class. Subsequently, we will apply this intermediate result
for showing that reachability can be witnessed by small partial
reachability trees. Finally, we put everything together in order to
obtain a polynomial-time algorithm.

\subsection{The Residue Reachability Problem}\label{sec:residue}
A cornerstone of our algorithm for reachability in \BVASSone\ is the
polynomial-time decidability of the following variant of the
reachability problem for \BVASSone:
\medskip
\problemx{Residue Reachability for \BVASSone}
{A \BVASSone\ $\B=(Q,\Delta,F)$, a configuration $q_0(n_0)$ and $d\geq 1$, 
  where $n_0$ and $d$ are given in unary.}
{Does there exist some $n\geq n_0$ such that $q_0(n)$ is reachable
and $n\equiv n_0\bmod d$?}
\medskip
The main result of this section is that residue reachability for
\BVASSone\ is decidable in polynomial time. Notice that setting $d=1$
allows for checking whether there exists some $n\ge n_0$ such that
$q(n)$ is reachable. We first introduce some auxiliary definitions
that allow us to abstract away concrete counter values of reachability
trees. A \emph{partial $d$-residue tree} is a binary tree $T\colon
U\rightarrow Q\times \Z_d$,
where each inner node $u\in U$ with
$T(u)=(q,n)$ satisfies precisely one of the following conditions:
\begin{enumerate}[(i)]
\item $u0,u1\in U$, and if $T(u0)=(p,m_0)$ and $T(u1)=(p',m_1)$ then
  $n\equiv m_0 + m_1 \bmod d$ and $(q,p,p')\in\Delta$;
\item $u0\in U, u1\not\in U$, and if $T(u0)=(p,m)$ then $m=n+z \bmod
  d$ and $(q,z,p)\in\Delta$.
\end{enumerate}
\medskip
We call a configuration from $Q\times \Z_d$ a \emph{residue
  configuration}. Given a set of configurations $S$, its
\emph{residue} is $S/\Z_d\defeq \{ (q,n\bmod d)\in Q\times \Z_d \mid q(n)\in S\}$. 
Likewise, given a partial reachability tree
$T\colon U\rightarrow Q\times\N$, the \emph{residue $T/\Z_d$ of $T$}
is $T/\Z_d\colon U\rightarrow Q\times \Z_d$, where $T/\Z_d(u)\defeq
T(u)/\Z_d$ for all $u\in U$. Clearly, $T/\Z_d$ is a partial residue tree. 

For the remainder of this section, fix some
\BVASSone\ $\B=(Q,\Delta,F)$, some configuration $q_0(n_0)$ and some
$d\geq 1$, where $n_0$ and $d$ are given in unary. In order to decide
residue reachability, one might be tempted to start with an initial
configuration and then to repeatedly apply transitions of $\B$ modulo
$d$ until the desired residue configuration is discovered. Such an
approach would, however, not be sound as it may lead to residue
configurations that, informally speaking, can only be obtained by
forcing the counter to drop below zero. Also, the simple alternative
of constructing a sufficiently large reachability tree is futile as it
may be of exponential size, cf.~\cref{sec:lower}. In order to balance
between those two extremes, we introduce reachability trees in which
all nodes except of the root are required to be bounded by some value
$j\in \N$: a partial reachability tree $T\colon U\rightarrow
Q\times\N$ is \emph{almost $j$-bounded} if $T(u)\in Q\times[0,j]$ for
all $u\in U\setminus\{\varepsilon\}$.  Note that every $j$-bounded
partial reachability tree is almost $j$-bounded.  The following
constant will be particularly useful:
\[
N\defeq |Q|\cdot d.
\]
Moreover, by $S$ we denote the set of configurations for which there
exists an $(n_0+N)$-bounded reachability tree and define for $i<j$:
\begin{align*}
  S  & \defeq  \{ (q,m)\in Q\times\N \mid q(m) \text{ has an } 
  (n_0+N)\text{-bounded reachability tree}\}\\
  S{[i,j]} & \defeq S \cap Q \times [i,j].
\end{align*}

\begin{restatable}{lemma}{lemComputationS}\label{lem:computationS}
  The set $S$ is computable in polynomial time.
\end{restatable}

For any set of residue configurations (modulo $d$) $V,W\subseteq
Q\times\Z_d$, we define the following sets that contain the result of
an application of a transition of $\B$ modulo $d$:
\begin{align*}
  \Delta(V) & \defeq\{(q,r-z\bmod{d})\mid (q,z,p)\in\Delta,(p,r)\in V\}\\
  \Delta(V,W) &
  \defeq\{(q,r_0+r_1\bmod{d})\mid(q,p_0,p_1)\in\Delta,(p_0,r_0)\in
  V,(p_1,r_1)\in W\}.
\end{align*}

Next, we inductively define a sequence of sets $R_i \subseteq Q\times
\Z_d$ for $i\geq 0$ whose fixed point will allow for deciding residue
reachability. The set $R_0$ consists of those pairs of control states
and residue classes that can be witnessed by a reachability tree that
is almost $(n_0+N)$-bounded and whose root has a counter value at
least $n_0+N$, and the $R_i$ for $i>0$ are obtained by application
of $\Delta:$
\begin{align*}
  R_0 & \defeq \{(q,n\bmod d)\in Q\times\Z_d\mid\\ & 
  ~~~~~~~~~~~~~~~~~ n\ge n_0+N, \text{
    $q(n)$ has an almost } (n_0+N)\text{-bounded reachability tree}\}\\
  R_{i+1} & \defeq
  R_i\cup
  \Delta(R_i)\cup\Delta(R_i,S/\Z_d)\cup\Delta(S/\Z_d,R_i)\cup\Delta(R_i,R_i).
\end{align*}
Since the cardinality of each $R_i$ is at most $N$, it
is easily seen that the sequence $(R_i)_{i\ge 0}$ reaches a fixed
point which can be computed in polynomial time.
\begin{restatable}{lemma}{lemFixedPointComputation}\label{lem:fixed-point-computation}
  The fixed point $R\defeq \bigcup_{i\geq 0}R_i$ equals $R_N$ and is
  computable in polynomial time.
\end{restatable}
In particular, $R$ together with $S$ yields the whole residue
reachability set.
\begin{restatable}{lemma}{lemModComputability}\label{lem:mod-computability}
  The set $X\defeq R\cup S[n_0,n_0+N]/\Z_d$ is computable in
  polynomial time. Moreover,
  \[
  X= \{(q,n\bmod d)\mid q\in Q,n\in\reach(q),n\geq n_0\}.
  \]
\end{restatable}
\begin{proof}[Proof (sketch)]
  Polynomial-time computability of $X$ follows immediately from
  \cref{lem:computationS,lem:fixed-point-computation}.  The proof of
  the stated equality is quite technical though not too difficult and
  \iftechrep{deferred to the appendix}{can be found in the technical
    report~\cite{GHLT16}}. The crucial part for the inclusion
  ``$\subseteq$'' is to show that for every $i\in[0,N]$ and each
  $(q,r)\in R_i$ there exists some $n\in\reach(q)$ with $n\geq
  n_0+N-i$ and $n\equiv r\bmod d$ by induction on $i$. For the
  converse inclusion the only interesting case is when a potential
  reachability tree $T$ is not $(n_0+N)$-bounded. One first shows that
  all $\prec$-maximal nodes $u$ in $T$ with $T(u)\not\in S$ satisfy
  $T(u)/\Z_d\in R_0$ and uses the fact that $\Delta(R,R)\subseteq R$
  and $\Delta(R)\subseteq R$ to conclude $T(\varepsilon)/\Z_d\in R$.
\end{proof}
The main result of this section now follows directly
from \cref{lem:mod-computability}.
\begin{theorem}\label{prop:residue}
  Residue reachability for \BVASSone\ is decidable in polynomial time.
\end{theorem}

\subsection{Expandable Partial Reachability Trees}\label{ssec:expandable}

We now employ our result on residue reachability to show that small
partial reachability trees suffice in order to witness
reachability. The key idea is to identify branches of partial
reachability trees that end in a leaf and which could, informally
speaking, be copied or pumped an arbitrary number of times, thus
achieving a counter value in the leaf that is large enough and lies in
a certain residue class of some modulus. Residue reachability then
witnesses that such a leaf could be completed in order to yield a
reachability tree. For the remainder of this section, fix some
\BVASSone\ $\B=(Q,\Delta,F)$.

Let us first introduce a couple of auxiliary definitions. Given a
partial reachability tree $T\colon U \to Q\times \N$ and $v,w\in
U$, the {\em lowest common ancestor} of $v,w\in U$ is defined as
\[
\lca(v,w)\defeq\max\{u\in U\mid u\preceq v\text{ and } u\preceq w\},
\] 
where the maximum is taken with respect to $\preceq$. Let $T(u)=q(n)$,
we define functions $\state(u)\defeq q$ and $\counter(u)\defeq n$ that
allow us to access the control state and the counter value at $u$,
respectively.
\begin{definition}
  A node $v\in U$ is \emph{increasing} if there is a proper
  ancestor $u\prec v$ such that $\state(u)=\state(v)$ and
  $\counter(u)<\counter(v)$; the maximal such $u$ is called the
  \emph{anchor} of $v$. 
We say that $T$ is \emph{exclusive} if the least common ancestor of any two
distinct increasing leaves is a proper ancestor of at least one of their
anchors.
Finally, we call $T$ \emph{expandable} if
  \begin{itemize}
  \item $T$ is exclusive,
  \item every leaf $v$ of $T$ is either accepting or an increasing
    leaf,
  \item every increasing leaf $v$ with anchor $u$ such that
    $T(v)=q(n)$ and $T(u)=q(m)$ induces a valid instance of the
    residue reachability problem, i.e., $q(l)$ is reachable for some
    $l\geq n$ and $l\equiv n \bmod (n-m)$.
  \end{itemize}
A node $u$ is said to be {\em exclusive} resp.\ {\em expandable}
if $T^{\downarrow u}$ is.
\end{definition}
Observe that nodes cannot be both accepting
and increasing because increasing nodes have strictly positive counter values
and accepting nodes must have counter value zero.
Exclusive and non-exclusive partial reachability trees are illustrated
in \cref{fig:exclusive}(a).
\begin{figure}
  \begin{center}
    \includegraphics[scale=1]{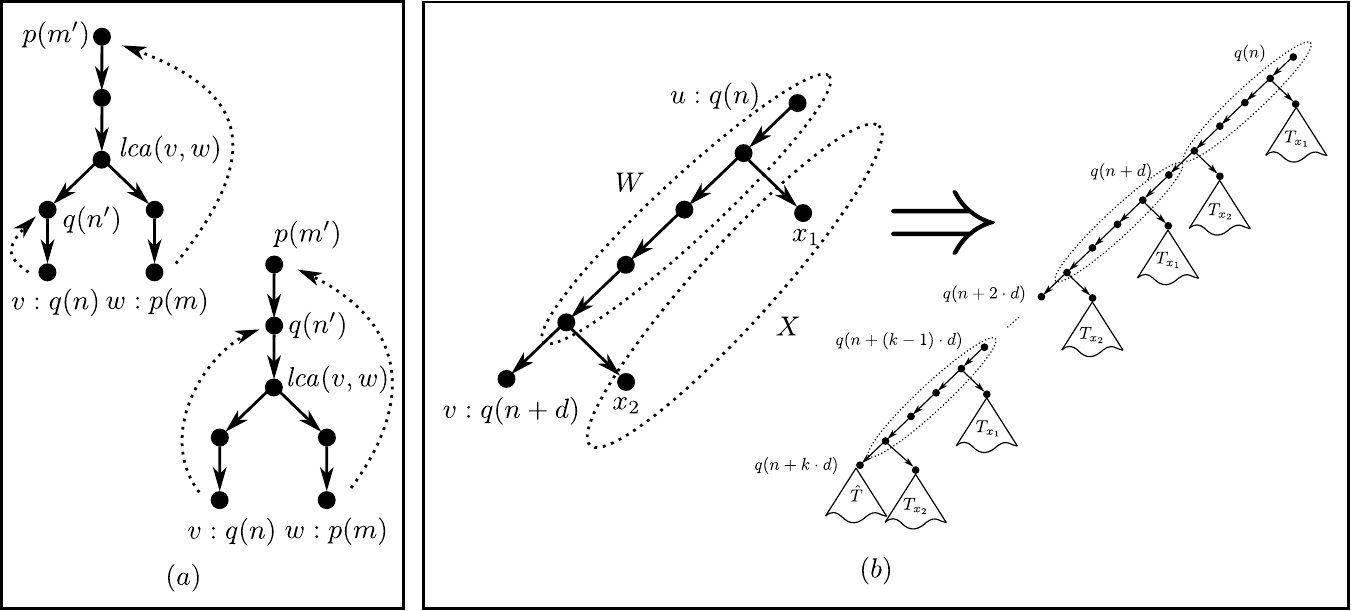}
  \end{center}
  \caption{(a) Illustration an exclusive (top) and a non-exclusive
    (bottom) partial reachability tree. Here, $v$ and $w$ are pumping
    nodes and anchor relationships are depicted as dashed arrows. (b)
    Illustration of the pumping argument in
    \cref{lem:expandable-reachable}.}
  \label{fig:exclusive}
\end{figure}
The next lemma states a useful fact that directly follows from the
pigeon-hole principle: whenever the counter increases on a branch by a
certain amount then the branch contains an increasing node and its
anchor.
\begin{restatable}{lemma}{lemPumpingNodes}\label{lem:pumping-nodes}
  Let $u$ and $v$ be nodes of a partial reachability tree such that
  $u\prec v $ and $\counter(u) + \abs{Q}\le \counter(v)$.  Then there
  exists an increasing node $v'$ with anchor $u'$ such that $u \preceq u'
  \prec v' \preceq v$.
\end{restatable}

The following lemma shows that every reachability tree gives rise to
an expandable reachability tree whose nodes have counter values
bounded polynomially in $|\B|$.
\begin{lemma}\label{lem:reachable-expandable}
  Suppose $q(n)$ is reachable and let $B\defeq 2\cdot \abs{Q}+n$. Then
  there exists an expandable $B$-bounded partial reachability tree
  with root $q(n)$.
\end{lemma}
\begin{proof}
  Let $T$ be a reachability tree with $T(\varepsilon)=q(n)$.  We call
  a node $w$ of $T$ {\em large} if $\counter(w)=B$. We obtain a
  partial reachability tree $T'$ from $T$ as follows. By
  \cref{lem:pumping-nodes}, every large node $w$ gives rise to at
  least one pair of nodes $(u,v)$ such that $u\prec v\preceq w$ and
  $v$ is an increasing node with anchor $u$. For every large node $w$
  that is minimal with respect to $\preceq$, we assign the maximal
  such pair $\pair(w)\defeq(u,v)$ with respect to the lexicographical
  ordering on nodes (more precisely, $(u,v)\preceq(u',v')$ if either,
  $u\prec u'$, or $u=u'$ and $v\preceq v'$).  Let $T'\colon
  U'\rightarrow Q\times\N$ denote the tree that one obtains from $T$
  by replacing all subtrees of $T$ that are rooted at some node $v$
  such that $\pair(w)=(u,v)$ for some minimal (with respect to
  $\preceq$) large node $w$ in $T$ by $\{v\}$ itself, i.e.\ such nodes
  $v$ become leaves. We now prove that $T'$ is $B$-bounded and
  exclusive:
  \begin{itemize}
  \item $T'$ is $B$-bounded since the $w$ above are chosen minimal
    with respect to $\preceq$ and hence $\counter(u)\le B$ for all nodes
  $u\in U'$.
  \item $T'$ is exclusive, which can be seen as follows. Striving
    for a contradiction, suppose that $T'$ is not exclusive. Then
    there are distinct increasing nodes $v,v'$ with anchors $u,u'$ such
    that $u,u' \preceq w \defeq \lca(v,v')$. Since
    $\counter(w)=\counter(w0) + \counter(w1) \le B$, we have
    $\counter(w0)\le B/2$ or $\counter(w1)\le B/2$, and assume without
    loss of generality that $\counter(w0)\le B/2$. 
    Since $B-B/2\ge \abs{Q}$, by \cref{lem:pumping-nodes} there is another
    increasing node $v''$ with anchor $u''$ such that $w0\preceq u'' \prec v''$,
    contradicting the assumed maximality of $(u,v)$.
  \item Every leaf is accepting or increasing, by definition of $T'$.
  \item Finally, every increasing leaf $u$ in $T'$ induces a positive
    residue-reachability instance. Since $T$ is a reachability tree,
    we have that $T(u)$ is reachable and thus $T'(u)$ is reachable.
    So in particular, it is reachable modulo $d=1$, i.e.\ if
    $T'(u)=q(n)$, then we can choose $(q(n),1)$ as the required valid
    instance of residue reachability.\qedhere
  \end{itemize}
\end{proof}

We now turn towards the converse direction and show that every
expandable tree witnesses reachability. We first state an auxiliary
lemma about structural properties of nodes in exclusive trees whose
proof can be found in the \iftechrep{appendix}{technical
  report~\cite{GHLT16}}.
\begin{restatable}{lemma}{lemWellNodes}\label{lem:well-nodes}
  For every node $u$ of an expandable partial reachability tree the
  following hold:
  \begin{enumerate}[(i)]
  \item If $u$ is the anchor of an increasing leaf $v$ then $u$ is
    expandable and all nodes $w$ such that $u\prec w\preceq v$ are
    not expandable.
  \item $u$ has at most one child that is not expandable.
  \end{enumerate}
\end{restatable}
\vspace{0.2cm} 
The previous lemma enables us to show that an expandable partial
reachability tree implies the existence of a reachability tree.

\begin{lemma}\label{lem:expandable-reachable}
  Let $T\colon U\to Q\times \N$ be an expandable partial reachability
  tree. Then for all $u\in U$, $T(u)$ is reachable or $u$ is not
  expandable.
\end{lemma}
\begin{proof}
  We prove the lemma by induction on $h(u)$. For the induction base,
  assume $h(u)=0$, hence $u$ is a leaf. Then $u$ is either accepting
  and thus $T(u)$ is reachable, or $u$ is not accepting and therefore
  an increasing leaf and so $T^{\downarrow u}$
  is not expandable by \cref{lem:well-nodes}(i).

  For the induction step, suppose $u$ is expandable. We distinguish two
  cases:
  \begin{itemize}
  \item {\em All children of $u$ are expandable.} We only treat the
    case when $u$ has two children, the case when $u$ has one child
    follows as a special case. Since the children $u0$ and $u1$ of $u$
    are expandable, by the induction hypothesis there are reachability
    trees $T_0\colon U_0\rightarrow Q\times\N$ and $T_1\colon
    U_1\rightarrow Q\times\N$ with $T_0(\varepsilon)=T(u0)$ and
    $T_1(\varepsilon)=T(u1)$.  We define the following tree $T_u\colon
    V\rightarrow Q\times\N$, where $V\defeq\{0\}U_0\cup\{1\}U_1\cup
    \{\varepsilon\}$, $T_u(\varepsilon)\defeq T(u)$ and $T_u(iv)\defeq
    T_i(v)$ for all $i\in\{0,1\}$.  Now $T_u$ is a reachability tree,
    hence $T_u(\varepsilon)=T(u)$ is reachable.
  \item {\em Some child of $u$ is not expandable.} For simplicity of
    presentation, let $u=\varepsilon$, the cases when
    $u\not=\varepsilon$ can be proven analogously. Moreover, let us
    assume that $T(u)=q(n)$. By \cref{lem:well-nodes}(ii) there
    is at most one such child, without loss of generality let $u0 = 0$
    be this child. Moreover, since $u$ is expandable and $u0$ is not
    expandable it must hold that $u$ is the anchor of some unique
    increasing leaf $v$, we may assume without loss of generality
    $v=u0^\ell$ for some $\ell\geq 1$. We must have $T(v)=q(n+d)$ for
    some $d\geq 1$.  Let $W=\{0^i\mid i\in[0,\ell-1]\}$ be the set all
    nodes in $T$ ``on the path from $u$ to $v$'' without $v$.  Let
    $X\defeq\{0^i1\in U\mid i\in[0,\ell-1] \}$ be the set of all right
    children of nodes in $W$.

    By \cref{lem:well-nodes}(i), all nodes in
    $\{0^i\mid i\in[1,\ell]\}$ are not expandable and consequently,
    \cref{lem:well-nodes}(ii) implies that all nodes in $X$ are
    expandable.  Hence by induction hypothesis, for every $x\in X$
    there is a reachability tree $T_x:U_x\rightarrow Q\times\N$ such that $T_x(\varepsilon) =
    T(x)$.      

    It remains to show that
    $T(u)=q(n)$ is reachable.  
    Since $T$ is expandable there exists some
    $m\geq n+d$ such that $q(m)$ is reachable and
    $m\equiv n\bmod{d}$.
    Let us assume $m=n+d+k\cdot d$ for some $k\geq 0$ 
    and let    
     $\widehat{T}\colon Z\rightarrow Q\times\N$ be a
    reachability tree for $q(m)$.

    We construct the following reachability tree $T'$ (formal
    definition below) for $q(n)$ as the tree one obtains from $T$ by
    replacing the leaf $v$ by the tree $T$ repeatedly exactly $k$
    times and by adding to the counter values of the resulting nodes
    from $0^*$ in the $i$-th copy the counter value $i\cdot d$. This
    procedure is illustrated in \cref{fig:exclusive}(b). Note that
    this process yields a partial reachability tree in which every
    leaf is accepting except for the leaf $0^{(k+1)\cdot\ell}$;
    therefore we replace this leaf by the tree
    $\widehat{T}\colon Z\rightarrow Q\times\N$.  Recall that
    $T_x\colon U_x\rightarrow Q\times\N$ is a reachability tree for
    $T(x)=T_x(\varepsilon)$.  Formally, we define
    $T'\colon \left(0^{(k+1)\cdot\ell}Z\cup\bigcup_{i=0}^{k} 0^{i\cdot\ell}
    (W\cup \bigcup\{xU_x\mid x\in X\})\right)\rightarrow Q\times\N$,
    where
\begin{itemize}
\item $T'(0^{(k+1)\cdot\ell}z)\defeq\widehat{T}(z)$ for all $z\in Z$,
\end{itemize}
and for all $i\in[0,k]$ we put
\begin{itemize}
\item $T'(0^{i\cdot\ell}w)\defeq i\cdot d+T(w)$ for all $w\in W$, and
\item $T'(0^{i\cdot\ell}xy)\defeq T_x(y)$ for all $x\in X$ and all $y\in U_x$.
\end{itemize}
It easily checked that the result is a
    reachability tree for 
    $T'(\varepsilon)=q(n)$. 
\qedhere
  \end{itemize}
\end{proof}
A consequence of the previous lemma is that in particular
$T(\varepsilon)$ is reachable for every expandable partial
reachability tree $T$. By combining
\cref{lem:reachable-expandable,lem:expandable-reachable}, we obtain the following
characterisation of reachability in \BVASSone.
\begin{proposition}\label{prop:expandable-reachable}
  A node $q(n)$ is reachable if, and only if, there exists an expandable
  $B$-bounded partial reachability tree $T$ with $T(\varepsilon)=q(n)$, where
  $B\defeq 2\cdot \abs{Q} + n$.
\end{proposition}

\subsection{The Algorithm}
In this section, we provide an alternating logspace procedure for
reachability in \BVASSone. This shows that reachability in
\BVASSone\ is decidable in deterministic polynomial time since
alternating logspace equals deterministic polynomial time
\cite{CKS81}. We employ the characterisation of reachability in
\BVASSone\ in terms of expandable $B$-bounded partial reachability of
\cref{prop:expandable-reachable}. First, by
\cref{prop:residue} we may assume the existence of an
alternating logspace procedure for residue reachability in \BVASSone,
i.e., an alternating logspace procedure
\textsc{ResidueReach}($q(n_0),d$) that has an accepting computation
if, and only if, $q(n)$ is reachable for some $n\geq n_0$ and $n\equiv
n_0\bmod d$. By application of this procedure, we show that one can
construct an alternating logspace procedure $\textsc{Reach}(q(n))$
that takes a configuration $q(n)$ as input and that has an accepting
computation if, and only if, there exists an expandable $B$-bounded
partial reachability tree $T$ with $T(\varepsilon)=q(n)$.

\algrenewcommand\algorithmiccomment[2][\footnotesize]{{#1\hfill\(\triangleright\) #2}}
\begin{algorithm}[t]
\caption{An alternating logspace procedure for reachability in \BVASSone.}
\footnotesize
\begin{algorithmic}[1]
\Procedure {Reach}{$q(n)$}
\If {$n\not\in[0,B]$}\textbf{ return false}\label{L Exceeds}
\EndIf
\If {$q(n)\in F\times\{0\}$}\ \textbf{return} \textbf{true}
\Else\  \textbf{non-deterministically guess} $t\in\Delta\cap(\{q\}\times Q\times Q\cup\{q\}\times\{-1,0,1\}\times Q)$
\If {$t=(q,p_1,p_2)\in Q^3$} 
\State \textbf{non-deterministically guess} $m_1,m_2\in[0,B]$ \text{s.t.} $n=m_1+m_2$
\State \textbf{return} $(\textsc{Reach}(p_1(m_1))$\textbf{ and }$\textsc{Reach}(p_2(m_2)))$
\State \hspace{0.7cm}\textbf{or} ($\textsc{AnchorReach}(q(n),p_1(m_1))$\textbf{ and }$\textsc{Reach}(p_2(m_2))$)
\State \hspace{0.7cm}\textbf{or} ($\textsc{AnchorReach}(q(n),p_2(m_2))$\textbf{ and }$\textsc{Reach}(p_1(m_1))$)
\Else\ \textbf{let } $t=(q,z,p)\in Q\times\{-1,0,1\}\times Q$
\State \textbf{return} $\textsc{Reach}(p(n+z))$ \textbf{or} $\textsc{AnchorReach}(q(n),p(n+z))$
\EndIf
\EndIf
\EndProcedure
\Statex
\Procedure {AnchorReach}{$q(n)$, $p(m)$}
\If {$\{n,m\}\not\subseteq[0,B]$}\textbf{ return false}
\EndIf
\If {$p=q$\textbf{ and }$m>n$\textbf{ and }\textsc{ResidueReach}$(q(n),m-n)$}\ \textbf{return} \textbf{true}
\Else\  \textbf{non-deterministically guess} $t\in\Delta\cap(\{p\}\times Q\times Q\cup\{p\}\times\{-1,0,1\}\times Q)$
\If {$t=(p,p_1,p_2)\in Q^3$} 
\State \textbf{non-deterministically guess} $m_1,m_2\in[0,B]$ \text{s.t.} $m=m_1+m_2$
\State \textbf{return} $\textsc{AnchorReach}(q(n),p_1(m_1))$\textbf{ and }$\textsc{Reach}(p_2(m_2))$
\State \hspace{0.7cm}\textbf{or} $\textsc{AnchorReach}(q(n),p_2(m_2))$\textbf{ and }$\textsc{Reach}(p_1(m_1))$
\Else\ \textbf{let } $t=(p,z,p')\in Q\times\{-1,0,1\}\times Q$
\State \textbf{return} $\textsc{AnchorReach}(q(n),p'(m+z))$
\EndIf
\EndIf
\EndProcedure
\end{algorithmic}
\label{alg:reachability}
\end{algorithm}

The idea is to simply to guess an expandable $B$-bounded partial
reachability tree $T$ in a top-down manner. The procedure
$\textsc{Reach}$ is defined above in
\cref{alg:reachability}. First in Line~2, $\textsc{Reach}$
rejects whenever the counter value $n$ is not in $[0,B]$ and accepts if
$q(n)$ is an accepting configuration (Line~3). 
Thus, subsequently we
may assume that $n\in[0,B]$. 
In Line~4, we non-deterministically choose
a transition $t\in\Delta$. 
If $t=(q,p_1,p_2)\in Q^3$ is a branching rule,
we non-deterministically guess how $n$ can be decomposed as
$n=m_1+m_2$. Moreover, we non-deterministically guess whether the
currently processed inner node of $T$ labelled by $q(n)$ will be an
anchor of some pumping leaf ``below.''
If not then we simply recursively call \textsc{Reach}($p_1(m_1)$) 
and \textsc{Reach}($p_2(m_2)$)
(Line~7). 
Otherwise, $q(n)$ will be the anchor of some pumping leaf
that is either in the subtree ``rooted at'' $p_1(m_1)$ (Line~8) or in the
subtree ``rooted at'' $p_2(m_2)$ (Line~9). Speaking in terms of
\cref{lem:well-nodes}, either the inner node corresponding to
configuration $p_1(m_1)$ is not exclusive or the one for $p_2(m_2)$ is not
exclusive. Suppose $p_1(m_1)$ is not exclusive, we then call a procedure
$\textsc{AnchorReach}(q(n),p(m_1))$ that takes two configurations as
arguments and tacitly assumes the first argument $q(n)$ is the
anchor and the second argument $p_1(m_1)$ corresponds to some inner node
that lies between the anchor and the pumping leaf it will eventually
correspond to.

In more detail, analogously to $\textsc{Reach}$ the procedure
$\textsc{AnchorReach}$ first checks whether the counter values of its
inputs both lie in $[0,B]$ (Line~13). 
If so it checks whether $p(m)$
corresponds to a valid pumping leaf of $q(n)$, i.e., it induces a
positive instance of the residue reachability problem by invoking
$\textsc{ResidueReach}(q(n),m-n)$ (Line~14). If not then a rule
$t\in\Delta$ is non-deterministically chosen (Line~15), and in case
$t$ is a branching rule, it is non-deterministically chosen which
``child'' of $p(m)$ is not exclusive, the other child is simply
checked for reachability by invoking procedure $\textsc{Reach}$
(Lines~18 and~19).

Obviously, $\textsc{Reach}$ and $\textsc{AnchorReach}$ can be
implemented in alternating logspace since the involved counter values
lie in the interval $[-1,B+1]$ and can hence be stored using a
logarithmic number of bits.

\section{Coverability and Boundedness}\label{sec:boundedness}
In this section, we show that the coverability and boundedness problem
for \BVASSone\ are also \P-complete.  The two problems are defined as
follows:
\medskip
\problemx{Coverability and Boundedness in \BVASSone} 
{A \BVASSone\ $\B=(Q,\Delta,F)$, a control state $q$ and $n\in \N$ encoded
  in unary.}
{\emph{Coverability:} Is there $m\ge n$ such that $q(m)$ is reachable? \\
\emph{Boundedness:} Is $\reach(q)$ finite?}
\medskip
If $q(n)$ is a positive instance of coverability then we call the
configuration $q(n)$ \emph{coverable}. A state $q$ is \emph{unbounded}
whenever $\reach(q)$ is unbounded (i.e.\ infinite).

Hardness for $\P$ is in both cases easily seen and similar to the
\P-hardness reduction from MCVP in \cref{prop:mcvp-construction}.

Moreover, the \P\ upper bound for coverability follows easily from the
\P\ upper bound for residue reachability since $q(n)$ is coverable if,
and only if, the pair $(q(n),1)$ is a positive instance of the 
residue reachability problem.
\begin{theorem}\label{thm:coverability}
  Coverability in \BVASSone\ is \P-complete.
\end{theorem}

The \P\ upper bound for boundedness, however, cannot be
derived immediately. In particular, as discussed in \cref{sec:lower}, there exists
a family of \BVASSone\ $(\B_n)_{n\ge 0}$ with some control state $q$
such that $\reach(q)$ is finite but of cardinality $2^n$.

For the remainder of this section, fix some
\BVASSone\ $\B=(Q,\Delta,F)$.
We first provide sufficient and necessary criteria that witness that a
control state is unbounded. Call a node $v$ in a reachability tree
\emph{decreasing} if there is an ancestor $u\prefix v$ with
$\state(u)=\state(v)$ and $\counter(u)>\counter(v)$.  The following
lemma, whose proof \iftechrep{is deferred to the appendix}{can be
  found in the technical report~\cite{GHLT16}}, shows that a
reachability tree that contains some decreasing node witnesses that
the control state at its root is unbounded.

\begin{restatable}{lemma}{lemUnboundednessNode}\label{lem:unboundedness-node}
    If a reachability tree $T$ with $T(\varepsilon)=q(n)$
    contains a decreasing node then $q$ is unbounded.
\end{restatable}
Conversely, the next lemma shows that a reachability tree whose root
is labelled with a configuration with a sufficiently large counter
value gives rise to a reachability tree which contains a decreasing
node, informally speaking, shortly below its root.

\begin{restatable}{lemma}{lemUnboundednessTrees}\label{lem:unboundedness-trees}
    Suppose $n>2^{\card{Q}}$ with $n\in\reach(q)$.  There exists a
    reachability tree $T\colon U\to Q\x\N$ for $q(n')$ where $n'\ge
    n$, and which contains a decreasing node $v$ with $\len{v}\le
    \card{Q}$.
\end{restatable}

A consequence of the two previous lemmas is that $q$ is unbounded if,
and only if, $\reach(q)$ contains some $n>2^{\card{Q}}$.
Even though the reachability trees in \cref{lem:unboundedness-trees}
are sufficient witnesses for unboundedness, they still contain much
more information than necessary and are potentially of exponential
size. In order to verify the existence of such a tree, exact counter
values and in fact the subtrees rooted in $v$ as well as all
incomparable nodes can be abstracted away, as shown in the lemma
below.

Let us write $\source{t}\eqdef
q$, $\targets{t}\eqdef \{p,p'\}$ and $\effect{t}\eqdef 0$, for the
\emph{source} and \emph{target states} and the \emph{effect} of a
branching transition $t=(q,p,p')\in Q^3$, respectively. Similarly, for
$t=(q,z,p)$ define $\source{t}\eqdef q$, $\targets{t}\eqdef \{p\}$ and
$\effect{t}\eqdef z$.

\begin{lemma}\label{lem:unboundedness-char}
  A control state $p_0$ is unbounded if, and only if, there is a
  sequence of control states and transitions $p_0 t_1 p_1 t_2 \cdots
  t_k p_k$ with $k\le \abs{Q}$ and some index $j<k$ such that
  \vspace{-0.2cm}
  \begin{enumerate}[(i)]
  \item $p_{i-1} = \source{t_i}$ and $p_i\in \targets{t_i}$ for all
    $1\le i\leq k$;
  \item $p_k=p_j$ and $p_i\neq p_j$ for all $0\leq i<j$;
  \item $p(0)$ is coverable for every
    $p\in\bigcup_{i=1}^{k}\targets{t_i}$; and
  \item for every $j< i \le k$, there exists $n_i\le \abs{Q}+1$ such that
      \begin{enumerate}
          \item if $t_i=(p_{i-1},p_i,p_i')\in Q^3$ or $t_i=(p_{i-1},p_i',p_i)\in Q^3$ then $p_i'(n_i)$ is
              coverable, else $n_i=0$,
          \item $\sum_{i=j+1}^k n_i > \sum_{i=j+1}^k \effect{t_i}$.
      \end{enumerate}
  \end{enumerate}
\end{lemma}
\vspace{0.2cm}
The last condition~(iv) expresses that the cyclic suffix
is consistent with the transition relation and guarantees a
decreasing node.
\begin{proof}
  If $p_0$ is unbounded, then by \cref{lem:unboundedness-trees} we can
  take a reachability tree $T$ containing a short decreasing node $v$,
  i.e., with $\len{v}\le \abs{Q}$.
  This decreasing node provides the claimed sequence:
  Conditions (i) and (ii) are immediate;
  for condition (iii) notice that for each mentioned state $p$
  some configuration $p(n)$ is reachable, as guaranteed by
  the respective subtree of $T$. This means in particular that
  $p(0)$ is coverable.

  For (iv), 
  first notice that the combined effect $\sum_{i=j+1}^k \effect{t_i}$
  of those transitions used between $v$ (where $\state(v)=p_k$) and its anchor 
  (with state $p_j=p_k$) is bounded by $\len{v}=k\le \abs{Q}$.
  Secondly, as for condition~(iii), we can assume that
  for all $p'_i$
  such that either 
  $t_i=(p_{i-1},p_i',p_i)\in Q^3$ or $t_i=(p_{i-1},p_i,p_i')\in Q^3$,
  some 
  configuration
  $p'_i(m_i)$ is reachable.
  For those $i\le k$ where $t_i\notin Q^3$, let $m_i\eqdef 0$.
  Now, for all $j<i\le k$, define
  $n_i\defeq\min\{|Q|+1,m_i\}$.
  
  Case~(iv)(a) holds immediately by definition of the $n_i$.
  To show Case~(iv)(b) we distinguish two cases. In case $m_i\geq
  |Q|+1$ for some such $i$ it follows that $n_i=|Q|+1$ and hence
  $\sum_{i=j+1}^k n_i\geq |Q|+1>\sum_{i=j+1}^k \effect{t_i}$.
  Otherwise, if all $m_i<|Q|+1$ then for all $i$ it holds that
  $n_i=m_i$ and so $\sum_{i=j+1}^k m_i\leq\sum_{i=j+1}^k \effect{t_i}$
  contradicts that $v$ is a decreasing node.

\iffalse
  Now again, observe that for any mentioned state $p'_i$,
  some configuration $p_i'(m_i)$ is reachable by assumption
  and so all $p_i'(n_i)$ for $n_i\le m_i$ are coverable.
  Suppose there do not exist small $n_i\le\abs{Q}$,
  such that the $p_i(n_i)$ are coverable,
  and such that their sum satisfies condition (iv).b.
  Then it must hold that the sum $\sum_{i=j+1}^k m_i$ itself is no larger than
  $\sum_{i=j+1}^k \effect{t_i}$.
  This however contradicts the assumption that $v$ is decreasing.
\fi  

  For the converse direction, assume a sequence as claimed
  above. Conditions~(i)-(iii) imply the existence of a reachability
  tree for some $p_0(n)$. Condition~(iv) ensures that there
  is such a tree with a decreasing node. We conclude by
  \cref{lem:unboundedness-node}.
\end{proof}

\Cref{lem:unboundedness-char} provides a
characterisation of unbounded states that directly translates into an
alternating logspace algorithm for the boundedness problem, similar to
\cref{alg:reachability},
which yields the \P\ upper bound.
In particular, observe that a witnessing sequence
satisfying Conditions~(i) and~(ii), as well as the numbers $n_i\le
\card{Q}+1$ can be guessed non-deterministically in logarithmic
space. Moreover, Conditions~(iii) and~(iv) are decidable in polynomial
time by \cref{thm:coverability}.

\begin{theorem}
  Boundedness in \BVASSone\ is \P-complete.
\end{theorem}

\iftechrep{
\section{Conclusion}
We showed that reachability, coverability and boundedness in
\BVASSone\ are all \P-complete and thereby established the first
decidability result for reachability in a subclass of \BVASS. This low
complexity is quite surprising since the general reachability problem
for \BVASS\ is at least non-elementary~\cite{LaS15} and there exist
families of instances of \BVASSone-reachability problems whose
reachability trees contain an exponential number of distinct counter
values, cf.~\cref{sec:lower}. The approach developed in this paper
shows that it is not necessary to explicitly construct a full
reachability tree in order to witness reachability. In fact, we showed
in \cref{sec:reachability} that the existence of so-called residue and
expandable reachability trees suffices in order to decide reachability
and can be witnessed in polynomial time.

Our approach is quite specific to having only one counter available in
\BVASSone\ and does not seem to immediately generalise to higher
dimensions. Nevertheless, we believe that this paper spreads some
optimism and provides sufficient evidence that obtaining results for
reachability in general \BVASS\ is not impossible.

}

\bibliographystyle{plainurl}
\bibliography{bibliography}

\iftechrep{
\newpage
\appendix
\section{Missing Proofs}
\subsection{Missing Proofs from \cref{sec:lower}}

An instance of \textsc{MCVP} is a Boolean circuit $\C$ consisting of
$n$ gates $g_1,\ldots,g_n$ such that for all $k\in[1,n]$ either
$g_k=\top$, $g_k=\bot$ or there are $1\leq i,j<k$ such that $g_k = g_i
\vee g_j$ or $g_k= g_i \wedge g_j$. \textsc{MCVP} is to decide whether
$\mathcal{C}$ evaluates to true, i.e.\ if $g_n$ evaluates to true.  We
note that \textsc{MCVP} is the canonical \P-complete
problem~\cite{GHR95}. The following proposition gives the lower bound
for \cref{thm:main}.

\propMCVPConstruction*
\begin{proof}
  From $\C$ we derive a \BVASSone\ $\B\defeq(Q,\Delta,F)$, where
  $Q\defeq \{ q_1,\ldots, q_n\}$, $F\defeq \{ q_i \mid g_i = \top \}$
  and $\Delta \defeq \{ (q_k,q_i,q_j) \mid g_k = g_i \wedge g_j \}
  \cup \{ (g_k,0,g_i), (g_k,0,g_j) \mid g_k = g_i \vee g_j \}$. Hence,
  $\wedge$-gates are simulated by splits and $\vee$-gates by
  non-deterministic branching. It is easily seen that $g_n$ evaluates
  to true if, and only if, $q_k(0)$ is reachable in $\B$.
\end{proof}

\propNPHardness*
\begin{proof}
  We first show that for any $m\in \N$ given in binary, we can in
  logarithmic space extend $\B_n$ constructed above with a control
  state $q_m$ such that $\reach(q_m)=\{ m \}$. Let $m=\sum_{0\le i \le
    n} b_i \cdot 2^i$ with $b_i\in \{0,1\}$ be the binary
  representation of $m$. We introduce additional fresh control states
  $q_m^i$, $0\le i\le n$, transitions $(q_m,0,q_m^n)$ and
  $(q_m^0,q_0,q_f)$, and for every $1\le i\le n$ transitions
    $(q_m^i,0,q_m^{i-1})$ if $b_i = 0$ and $(q_m^i,q_i,q_m^{i-1})$ if
  $b_i = 1$.
  It is easily checked that $\reach(q_m)=\{ m \}$.

  In order to show hardness for \NP, we reduce from the problem
  \textsc{Subset Sum}. Given a finite set $S=\{ m_1,\ldots,
  m_k\}\subseteq \N$ and $t\in \N$ with all numbers encoded in binary,
  \textsc{Subset Sum} is the problem to decide whether there are
  $c_1,\ldots,c_k\in \{0,1\}$ such that $t=\sum_{1\le i \le k}
  c_i\cdot m_i$ and is known to be \NP-complete~\cite{Pap94}. As shown
  above, we can construct a \BVASS\ $\B$ with control states
  $q_{m_1},\ldots, q_{m_k}$ such that $\reach(q_{m_i})=\{ m_i \}$. We
  introduce additional fresh control states $q_{c_1},\ldots, q_{c_k}$
  that allow us to non-deterministically make a choice for every $c_i$
  by introducing for every $1\le i < k$ transitions
  $(q_{c_i},0,q_{c_{i+1}})$ and $(q_{c_i},q_{m_i},q_{c_{i+1}})$. It is
  now easily seen that the instance $(S,t)$ of \textsc{Subset Sum} is
  valid if, and only if, $q_{c_1}(t)$ is reachable.
\end{proof}

\corNPBVASSTwo
\begin{proof}[Proof (sketch)]
  The statement follows from an easy adaption of the proof of
  \cref{prop:np-hardness}. It suffices to show how to construct a
  \BVASS{2} that reaches the control state $q_{c_1}$ from
  \cref{prop:np-hardness} with counter values $(t,0)$. But this
  can easily be achieved by first adding a non-deterministic number of
  times $(1,1)$ to the counter and then by branching into the control
  states $q_{c_1}$ and $q_t$, where $q_t$ is suitably adjusted such that
  $\reach(q_t)=\{ (0,t) \}$.
\end{proof}

\newpage
\subsection{Missing Proofs from \cref{sec:residue}}

\lemComputationS*
\begin{proof}
We note that $N$ is polynomially bounded in $\card{\B}+\card{d}$.
Moreover, $S\subseteq Q\times[0,n+N]$ and $S$ can be computed in
polynomial time by using a dynamic programming approach.
\end{proof}

\lemFixedPointComputation*
\begin{proof}
  Analogously to the computation of $S$ in
  \cref{lem:computationS}, one shows that $R_0$ is computable in
  polynomial time. To see that $R=R_N$, note that by definition we
  have $R_i\subseteq R_{i+1}$ for all $i\in\N$. If $R_i\subset
  R_{i+1}$, there is at least one pair from $Q\times \Z_d$ that is in
  $R_{i+1}$ and not in $R_i$. Since there are at most $N$ such pairs,
  the sequence stabilises after at most $N$ steps at $R_N$.  Since $N$
  is polynomial in $\abs{\B}+d$, consequently $R_N$ can also be
  computed in polynomial time.
\end{proof}

\lemModComputability*
\begin{proof}
  Polynomial-time computability of $X$ follows immediately from the
  polynomial time computability of $S$ (\cref{lem:computationS})
  and of $R$ (\cref{lem:fixed-point-computation}). It thus
  remains to prove that $X=\{(q,n\bmod d)\mid q\in
  Q,n\in\reach(q),n\geq n_0\}$.

(``$\subseteq$'') Trivially, $S[n_0,n_0+N]/\Z_d\subseteq \{(q,n\bmod
d)\mid q\in Q,n\in\reach(q),n\geq n_0\}$ since $S\subseteq \{(q,n)\mid
q\in Q,n\in\reach(q)\}$.  Hence it remains to show that $R$ is
contained in $\{(q,n\bmod d)\mid q\in Q,n\in\reach(q),n\geq n_0\}$.

To prove this, we show that for every $i\in[0,N]$ and each $(q,r)\in
R_i$ there exists some $n\in\reach(q)$ with $n\geq n_0+N-i$ and
$n\equiv r\bmod d$ by induction on $i$. We note that this is
sufficient to prove since $R=R_N$ and thus for each $(q,r)\in R$ there
exists some $n\in\reach(q)$ with $n\geq n_0$ and $n\equiv r\bmod{d}$.

For the induction base, i.e.\ $i=0$, we recall that for each $(q,r)\in
R_0$ there exists some $n\geq n_0+N=n_0+N-i$ such that $n\equiv r\bmod
d$ and there is some almost $(n_0+N)$-bounded reachability tree whose
root is labelled with $q(n)$ by definition of $R_0$; in particular
$n\in\reach(q)$.

For the induction step, let $i+1\leq N$ and let us assume $(q,r)\in
R_{i+1}$.  If already $(q,r)\in R_{i}$ then $(q,r)$ satisfies the
desired property immediately by the induction hypothesis. Otherwise, if
$(q,r)\in\Delta(R_i)$ then $r\equiv r'-z\bmod d$ for some $(q',r')\in
R_i$ and some $(q,z,q')\in\Delta$.  By the induction hypothesis, there
exists some $n'\in\reach(q')$ with $n'\geq n_0+N-i$ and $n'\equiv
r'\bmod d$.  For $n=n'-z$, we have $n\equiv r'-z\equiv r\bmod{d}$ and
since $n\geq n'-|z|\geq n_0+N-i-|z|\geq n_0+N-(i+1)\geq n_0\geq0$ it
follows $n\in\reach(q)$.

It remains to consider the case when
$(q,r)\in\Delta(R_i,S/\Z_d)\cup\Delta(S/\Z_d,R_i)\cup\Delta(R_i,R_i)$.
We only treat the case $(q,r)\in\Delta(R_i,S/\Z_d)$, the other cases
can be proven analogously.  In this case we have $r\equiv
r'-n''\bmod{d}$ for some $(p',r')\in R_i$ and some $(p'',n'')\in S$,
where $(q,p',p'')\in\Delta$. Clearly, $n''\in\reach(p'')$ by
definition of $S$.  By the induction hypothesis, there exists some
$n'\geq n_0+N-i$ such that $n'\in\reach(p')$ and $n'\equiv
n\bmod{d}$. Let $n=n'+n''$. Hence $n\in\reach(q)$, since
$n'\in\reach(p')$ and $n''\in\reach(p'')$. Obviously $n\equiv
r\bmod{d}$ and, finally, $n\geq n'\geq n_0+N-i\geq n_0+N-(i+1)$.

(``$\supseteq$'') Assume some $q(n)$ is reachable for some $n\geq
n_0$. We prove that $(q,n\bmod{d})\in X$. To this end, let us fix some
reachability tree $T\colon U\rightarrow Q\times\N$ for $q(n)$. If $T$
is $(n_0+N)$-bounded it follows that $(q,n)\in S$ and we are done
since $n_0\leq n\leq n_0+N$.

Consequently, let us assume that $T$ is not $(n_0+N)$-bounded. First,
observe that $T(u)\in F\times\{0\}\subseteq S$ for all leaves $u\in
U$. In addition, the set of nodes $V\defeq\{u\in U\mid T(u)\not\in
S\}$ is non-empty for otherwise $T$ would be
$(n_0+N)$-bounded. Moreover, $V$ is prefix-closed and note that every
$\preceq$-maximal node $v$ in $V$ satisfies $T(v)/\Z_d\in R_0$ by the
choice of $V$. For every node $v\in V$ that is not $\preceq$-maximal,
we either have $T(v)/\Z_d\in\Delta(T(v0)/\Z_d)$ (if $v0$ is the only
child of $v$) or $T(v)/\Z_d\in\Delta(T(v0)/\Z_d,T(v1)/\Z_d)$ (if $v$
has two children $v0$ and $v1$). Moreover, note that
$\Delta(R)\subseteq R$ and $\Delta(R,R)\subseteq R$. This shows that
$T/\Z_d(V)$ is contained in $R$, in particular $T/\Z_d(\varepsilon)\in
R$ which proves $(q,n\bmod d)\in R\subseteq X$.
\end{proof}

\subsection{Missing Proofs from \cref{ssec:expandable}}

\lemPumpingNodes*
\begin{proof}
  The counter value of a node exceeds that of its parent by at most one.
  Consequently, for every $\counter(u)\leq i \leq \counter(v)$ there is
  a node $u\preceq u_i\preceq v$ with $\counter(u_i)=i$
  and further,
  for all $j<k$ in between $\counter(u)$ and $\counter(v)$ it holds that $u_j\prec u_k$.
  Since $\counter(v)-\counter(u)\geq|Q|$,
  there must be some $\counter(u)\leq j<k\leq\counter(v)$ such that
  $\state(u_j) = \state(u_k)$.
  Then $u'\eqdef u_j$ and $v'\eqdef u_k$ satisfy the claim.
\end{proof}

\lemWellNodes*
\begin{proof}
  Regarding~(i), every anchor $u'$ of every other increasing leaf $v'$
  of $T^{\downarrow u}$ is a strict descendant of $u$, since otherwise
  $u,u' \preceq \lca(v,v')$, contradicting $T$ being
  exclusive. Consequently, $T$ being exclusive implies $u$ being
  exclusive. Moreover, for $w$ such that $u\prec w\preceq v$, due to
  an anchor being maximal, $v$ has no anchor on the subtree rooted at
  any such $w$. Hence, $w$ violates the condition of being expandable.

  Regarding~(ii), suppose $u0$ and $u1$ are both not
  expandable. This can only be if both are not exclusive,
  as $T$ is expandable. If all increasing leaves of $T^{\downarrow u0}$ had
  their anchors in $T^{\downarrow u0}$, then $u0$ would be exclusive. Hence there
  is some increasing leaf $v_0$ in $T^{\downarrow u0}$ with anchor $u_0$ such
  that $u_0\preceq u$. Likewise, we find an increasing leaf $v_1$ with
  anchor $u_1$ such that $u_1\preceq u$ in $T^{\downarrow u1}$.
  But then $u_0,u_1\preceq \lca(v_0,v_1)$ and hence $T$ is not
  exclusive, a contradiction.
\end{proof}

\newpage
\subsection{Missing Proofs from \cref{sec:boundedness}}

\lemUnboundednessNode*
\begin{proof}
  It suffices to observe that one can unfold the cyclic suffix of a
  decreasing node $v$ by replacing the subtree rooted in $v$ by that
  one rooted in $u\prefix v$. This construction is analogous to the
  construction in the proof of \cref{lem:expandable-reachable}, with
  the only difference that the effect of the cycle is negative
  here. The result of such an operation is a reachability tree whose
  root is labelled with a configuration that has the same control
  state and whose counter value is strictly increased. Moreover, this
  reachability tree still contains a decreasing node. Such an
  unfolding can therefore be repeated arbitrarily often, from which
  the claim follows.
\end{proof}

\lemUnboundednessTrees*
\begin{proof}
  In any reachability tree it is possible to collapse the part between
  two nodes $u\prefix v$ if $\state(u)=\state(v)$ and
  $\counter(u)\le\counter(v)$, that is, to replace the subtree rooted
  in $u$ by the one rooted in $v$. The result of this is a
  reachability tree with fewer nodes and where the root has the same
  state and a counter value at least as large as in the original tree.

  Thus, we may assume with no loss of generality a reachability tree
  $T$ with root $T(\eps)=q(n)$ for $n\ge2^{\card{Q}}$ and such that
  for any two nodes $u\prefix v$ with $\state(u)=\state(v)$, it holds
  that $\counter(u)> \counter(v)$.

  In order to find a decreasing node, we move from the root downwards,
  always choosing the successor with the largest counter value. This
  way, the counter value of a chosen node is at least half as large as
  the counter of its parent.  Since the value in the root is greater
  or equal to $2^{\card{Q}}$, this means that the produced sequence is
  longer than $\card{Q}$. In particular, the prefix of length
  $\card{Q}$ must contain a decreasing node.
\end{proof}

}{}

\end{document}